\documentclass[runningheads]{llncs}
\usepackage[T1]{fontenc}
\usepackage{graphicx}
\usepackage{amsfonts}
\usepackage{amsmath}
\usepackage{braket}
\usepackage{tikz}
\usepackage{tikz-3dplot}
\usetikzlibrary{arrows,calc,backgrounds,3d}

\newcommand{\MC}{\texttt{MAX-CUT}}
\begin{document}
\title{Warm-Started QAOA with Aligned Mixers Converges Slowly Near the Poles of the Bloch Sphere}
%
\titlerunning{Warm-Started QAOA Converges Slowly Near the Poles of the Bloch Sphere}
%
\author{Reuben Tate\orcidID{0000-0002-9170-8906} \and
Stephan Eidenbenz\orcidID{0000-0002-2628-1854} }
\authorrunning{R. Tate, S. Eidenbenz}
%
\institute{CCS-3: Information Sciences, Los Alamos National Laboratory, USA 
 \\
\email{\{rtate, eidenben\}@lanl.gov}}
\maketitle              
\begin{abstract}
In order to boost the performance of the Quantum Approximate Optimization Algorithm (QAOA) to solve problems in combinatorial optimization, researchers have leveraged the solutions returned from classical algorithms in order to create a warm-started quantum initial state for QAOA that is biased towards "good" solutions. Cain et al. showed that if the classically-obtained solutions are mapped to the poles of the Bloch sphere, then vanilla QAOA with the standard mixer "gets stuck". If the classically-obtained solution is instead mapped to within some angle $\theta$ from the poles of the Bloch sphere, creating an initial product state, then QAOA with optimal variational parameters is known to converge to the optimal solution with increased circuit depth if the mixer is modified to be "aligned" with the warm-start initial state. Leveraging recent work of Benchasattabuse et al., we provide theoretical lower bounds on the circuit depth necessary for this form of warm-started QAOA to achieve a desired change $\Delta \lambda$ in approximation ratio; in particular, we show that for small $\theta$, the lower bound on the circuit depth roughly scales proportionally with $\Delta \lambda/\theta$. 

\end{abstract}
\section{Introduction}

The Quantum Approximate Optimization Algorithm (QAOA) \cite{farhi2014quantum} and its greatly generalized sibling Quantum Alternating Operator Ansatz (also QAOA) \cite{HWORVB17} is a leading quantum algorithm or perhaps a quantum heuristic to find approximate solutions to combinatorial optimization problems, such as Traveling Salesperson, Satisfiability, or -- and this is the most commonly studied case -- Maximum Cut. While a diversity of QAOA variations have emerged over the past decade \cite{blekos2024review}, the concept of using a quantum state representation of a high-quality classically obtained approximate solution as starting state for QAOA is our focus for analysis. More precisely, we study a case where the classically-obtained initial solution is mapped to within some angle $\theta$ from the poles of the Bloch sphere, creating an initial product state; this QAOA variant with optimal variational parameters is known to converge to the optimal solution with increased circuit depth if the mixer is modified to be ``aligned" with the warm-start initial state.

Provable approximation ratio results for QAOA are still very rare. Our work focuses on the effect that the parameter $\theta$ has on the behavior of the QAOA circuit, thus giving new insights into the theory of QAOAs. More specifically, we provide $\theta$-dependent lower bounds on the required circuit depth needed to achieve a desired change in approximation ratio. The key takeaway is the following: \emph{In the context of warm-started QAOA (with ``aligned" mixers), one should avoid initializing the warm-start very close to poles of the Bloch sphere (i.e. $\theta$ near 0).} Some preliminary numerical simulations \cite{Tate2023warmstartedqaoa,egger2021warm} had already suggested that initializing near the poles was not ideal; this work solidifies such a result with a theoretical backing.  While our results do not show a worst-case advantage of warm-start QAOA over either the traditional QAOA approach (e.g. approximation ratio of $0.6924$ for 3-regular \MC{} \cite{farhi2014quantum}), 
the practical advantage of warm-started QAOA in NISQ and numerical experiments should be explored further.

This paper is organized as follows: in Section \ref{sec:notationBackground}, we set up the necessary notation for this work and briefly review the standard QAOA algorithm, in Section \ref{sec:initialStates}, we give a more detailed description of the initial warm-start used and the corresponding ``aligned" mixer, in Section \ref{sec:toyProblem}, we consider an illustrative simple toy problem that provides a geometric intuition for why one should expect initializations near the poles to perform poorly, and in Section \ref{sec:boundsOnRequiredDepth}, we provide a theoretical analysis that shows that such an intuition holds more generally as well by proving lower bounds on the circuit depth. We provide a discussion and conclude in Section \ref{sec:discussionConclusion}.

\section{Notation and Background}
\label{sec:notationBackground}
For the purposes of this work, we assume throughout that we are working with some classical objective function $c: \{0,1\}^n \to \mathbb{Z}^{\geq 0}$ to be maximized, which maps bitstrings of a particular length to a non-negative integer. We also assume that $c$ is not identically zero everywhere; together with non-negativity, this ensures that $c_\text{max} := \max_{x \in \{0,1\}^n} c(x) > 0$. Let $\mathcal{A}$ be some (possibly randomized) algorithm and let $\mathcal{A}(c)$ be the (expected) objective value obtained by running $\mathcal{A}$ with an instance corresponding to objective function $c$, we then define the corresponding (instance-specific) approximation ratio\footnote{In the classical optimization literature, the term \emph{approximation ratio} is often only used to refer to \emph{worst-case} ratio (of the expected algorithm objective value to the optimal objective value) amongst some class of objective functions. In recent years, the term has also been used to refer to the ratio for individual instances, which is how the term will be used in this work.} as, $\lambda_\mathcal{A}(c) = \frac{\mathcal{A}(c)}{c_\text{max}}$; when the context is clear, we will often write $\lambda_\mathcal{A}(c)$ as $\lambda_\mathcal{A}$ or even just $\lambda$.

We will often use the \MC{} problem as an example which is as follows: given a graph $G=(V,E)$, partition the vertices $V$ into two groups, $S$ and $V \setminus S$, so that the number of edges between the groups is maximized. If $|V|=n$, the partitioning of the vertices can be expressed as a bitstring $x$ whose $i$th bit is 0 if the $i$th vertex is in $S$ and 1 otherwise. For any fixed choice of graph $G$, the corresponding objective function for this problem is: $c(x) = \sum_{(i,j) \in E} \frac{1}{2}(1-z_iz_j)$ where $z_i = 2x_i-1 \in \{-1,+1\}$ for all $i$.
\subsection{Standard QAOA}
\label{sec:QAOA}
We next review the QAOA algorithm and set up the needed notation that will be used throughout this work. We use $X,Y,Z$ to denote the standard Pauli matrices.
For a multi-qubit system, we use $X_j, Y_j, Z_j$ to denote the operation of applying $X,Y,$ or $Z$ to the $j$th qubit respectively. We use $I$ and $\mathbf{0}$ to denote the identity and all-zeros matrix respectively; the dimensions of such matrices will be clear from context. For any square matrix $M$ and scalar $t$, we define the corresponding matrix
$U(M, c) := e^{-itM}$.
For a combinatorial optimization problem determined by a cost function $c: \{0,1\}^n \to \mathbb{Z}^{\geq 0}$ on $n$-length bitstrings, we define the corresponding cost Hamiltonian as the matrix $C$ such that $C\ket{b} = c(b)\ket{b}$.

Given a classical objective function $c$ on $n$-length bitstrings (with corresponding cost Hamiltonian $C$), a Hermitian matrix $B$ (called the mixing Hamiltonian) of appropriate size, and initial quantum state $\ket{\psi_i}$ in a $2^n$-dimensional Hilbert space, a circuit depth $p$, and variational parameters $\gamma = (\gamma_1, \dots, \gamma_p), \beta = (\beta_1, \dots, \beta_p)$, we define the following variational waveform of depth-$p$ QAOA as follows:
\begin{equation}\label{eqn:waveform}\ket{\psi_p(\gamma,\beta)} := U(B, \beta_p)U(C, \gamma_p) \cdots U(B, \beta_1)U(C, \gamma_1)\ket{\psi_i}.\end{equation}


For unconstrained optimization problems, the mixing Hamiltonian $B$ is usually taken to be the transverse field mixer, which, for an $n$-qubit system, is defined as
\begin{equation}B_\text{TF} = \sum_{j=1}^n X_j.\end{equation}
Additionally, the starting state is usually taken to be an equal superposition of all $2^n$ bitstrings of length $n$:
\begin{equation}\ket{\psi_i} = \ket{+}^{\otimes n} = \frac{1}{\sqrt{2^n}}\sum_{b \in \{0,1\}^n} \ket{b}.\end{equation}

In the case where $\ket{\psi_i}$ is the most-excited state of $H_B$, there exists a choice of angles $\gamma$ and $\beta$ for which the QAOA circuit can be viewed as a Trotterization of the Quantum Adiabatic Algorithm which is known to, under mild assumptions (see \cite{farhi2014quantum,binkowski2023elementary}), converge to the optimal solution given enough time; in other words, for a cost function $c$ that we wish to maximize, we have that:
$$\lim_{p \to \infty} \max_{\gamma,\beta} \Big[ \bra{\psi_p(\gamma,\beta)}C\ket{\psi_p(\gamma,\beta)}\Big] =  \max_{x \in \{0,1\}^n} c(x),$$
where $\bra{\psi_p(\gamma,\beta)}C\ket{\psi_p(\gamma,\beta)}$ is the expected cost value obtained from measuring the output state of QAOA.


\section{Initial Product States and Aligned Mixers}
\label{sec:initialStates}
For the standard QAOA algorithm \cite{farhi2014quantum} and many of its variants, the equal superposition $\ket{\psi_i} = \ket{+}^{\otimes n}$ is used as the initial state. In the context of \MC{} for example, quantum measurement of $\ket{+}^{\otimes n}$ produces a uniform distribution of all $2^n$ cuts in the graph; put another way, each vertex, independent of the other vertices, has probability $1/2$ of being on one side of the cut or the other.

However, one can consider modifying the QAOA algorithm by using a different initial state $\ket{\psi_i}$. Often, such initial states are constructed as a function of \emph{classically} obtained solutions. This method is parametrized by a parameter $\theta$ which we refer to as the \emph{initialization angle}; to this end, we first introduce some helpful notation.

\subsection{Construction of Warm-Started States}
\label{sec:constructionOfWarmStartedStates}
Let $\vec{n} = (x,y,z)$ be a unit vector written in Cartesian coordinates. We let $\ket{\vec{n}}$ denote the single-qubit quantum state whose qubit position on the Bloch sphere is $\vec{n}$. For $\vec{n}=(x,y,z)$, we define the following single-qubit operation:
\begin{equation}B_{\vec{n}} = xX+yY+zZ,\end{equation}
and let $B_{\vec{n},j}$ denote the operation of applying the operation $B_{\vec{n}}$ on the $j$th qubit. The unitary $U(B_{\vec{n},j}, \beta)$ can be geometrically interpreted as a single-qubit rotation by angle $2\beta$ about the axis that points in the $\vec{n}$ direction \cite{blochSphereRotations}.

For an $n$-qubit product state $\ket{s} = \bigotimes_{j=1}^n \ket{\vec{n}_j}$, we define an $n$-qubit operation $B_{\ket{s}}$ in terms of the single-qubit operations above:
\begin{equation}B_{\ket{s}} =  \sum_{j=1}^n B_{\vec{n_j}, j}.\end{equation}

When $\ket{s}$ is a product state, one can show that $\ket{s}$ can be prepared and that $B_{\ket{s}}$ can be implemented in most quantum devices with a constant-depth circuit using standard single-qubit rotation gates about the $x,y,$ and $z$ axes. Moreover, they \cite{Tate2023warmstartedqaoa} remark that $\ket{s}$ is a ground state of $B_{\ket{s}}$ for any product state $\ket{s}$. Thus, as remarked in \cite{Tate2023warmstartedqaoa}, as long as none of the qubits of the initial product state $\ket{\psi_i}$ are initialized at the poles, running QAOA with the standard phase separator $C$ and mixer $B_{\ket{\psi_i}}$ will yield the optimal solution as the circuit depth goes to infinity (assuming that $\gamma$ and $\beta$ are chosen optimally) \cite{Tate2023warmstartedqaoa}. In general, whenever the initial state of QAOA is the ground state of the mixer, we say that the mixer is \emph{aligned} with the initial state, and refer to this category of QAOA variants as \emph{QAOA with aligned mixers}.



Prior approaches for warm-started QAOA considered initial states of the form $\ket{\psi_i} = \bigotimes_{j=1}^n \ket{\vec{n}_j}$ where $\vec{n}_1,\dots, \vec{n}_n$ are obtained by some classical procedure. In the work by Egger et al. \cite{egger2021warm}, for problems whose corresponding QUBO (Quadratic Unconstrained Binary Optimization) formulation satisfies certain properties, they solve a relaxation of the QUBO and map the solutions to states in,
\begin{equation}\label{eqn:arc}\textbf{Arc} = \{\cos(\theta/2)\ket{0}+\sin(\theta/2)\ket{1} : \theta \in (0, \pi)\},\end{equation}
i.e., points on the Bloch sphere that intersect with the $xz$-plane with non-negative $x$ coordinate; the blue arc in Figure \ref{fig:singleCutInitialization} corresponds to the possible qubit positions. For the \MC{} problem, Tate et al. \cite{Tate2023warmstartedqaoa,tate2023bridging} consider higher-dimensional relaxations of the Max-Cut problem (i.e. the Burer-Monteiro relaxation and the relaxation used in the Goemans-Williamson algorithm) which, after a possible projection, yield points all over the surface of the Bloch sphere; however, any advantages this method has over others cannot be attributed to simply utilizing more of the Bloch sphere's surface as seen in the following remark.

\begin{remark}
\label{remark:zeroPhase}
    Tate et al. \cite{Tate2023warmstartedqaoa} show that for every initial product state $\ket{\psi_i} = \bigotimes_{j=1}^n \ket{\vec{n}_j}$, there exists a different initial product state $\ket{\psi_i'} = \bigotimes_{j=1}^n \ket{\vec{n}_j'}$ with $\vec{n}_j' \in \textbf{Arc}$ for all $j$, such that, up to a global phase, QAOA with aligned mixers and initial state $\ket{\psi_i}$ returns the same state as QAOA with aligned mixers and initial state $\ket{\psi_i'}$. This property holds for any optimization problem and corresponding cost Hamiltonian $C$. More specifically, if $\ket{\vec{n}_j} = \cos(\theta_j/2)\ket{0} + e^{i\phi_j}\sin(\theta_j/2)\ket{1}$ is an arbitrary qubit on the Bloch sphere with polar angle $0 \leq \theta_j \leq \pi$ and azimuthal angle $0 \leq \phi_j < 2\pi$, then one can choose $\ket{\vec{n}_j'} = \cos(\theta_j/2)\ket{0} + \sin(\theta_j/2)\ket{1}$.
\end{remark}

We now explicitly define the warm-starts used in this work. First, for any product state $\ket{s} = \bigotimes_{j=1}^n \ket{s_j}$, we define $\theta_j$ and $\phi_j$ to be the polar and azimuthal angle of $\ket{s_j}$ on the Bloch sphere, i.e., $\ket{s_j} = \cos(\theta/2)\ket{0}+e^{i\phi}\sin(\theta/2)\ket{1}$. For each $j \in [n]$, we use $\hat{\theta_j}$ to measure the angle that the $j$th qubit is to the nearest pole, i.e., $\hat{\theta_j} = \min(\theta_j, \pi - \theta_j)$. This work focuses on the restriction that each qubit is at most angle $\theta \in [0,\pi/2]$ away from one of the poles, i.e., $\hat{\theta_j} \leq \theta$ for all $j\in [n]$; we refer to such warm-starts as \emph{within-$\theta$-warm-starts}. In light of Remark \ref{remark:zeroPhase} above, in the context of warm-started QAOA with within-$\theta$-warm-starts with aligned mixers, it suffices to only consider warm-starts of the form $\ket{\psi_i} = \bigotimes_{j=1}^n \ket{\theta_j}$ where $\ket{\theta_j} := \cos(\theta/2)\ket{0}+\sin(\theta/2)\ket{1} \in \textbf{Arc}$ with $\hat{\theta_j} \leq \theta$ for all $j\in [n]$.

In addition, we also consider a subclass of within-$\theta$-warm-starts where each qubit is \emph{exactly} some fixed angle from the poles of the Bloch sphere, i.e., $\theta_j = \theta$ for all $j \in [n]$; we refer to such warm-starts as \emph{at-$\theta$-warm-starts}. As seen in Figure \ref{fig:singleCutInitialization}, this forces each qubit to be in one of two positions on the Bloch sphere; this naturally induces a \emph{corresponding bitstring} $b\in\{0,1\}^n$ with the property that at $\theta = 0$, measuring the warm-start yields exactly the state $\ket{b}$.

\begin{figure}[t]
\centering
    \pgfmathsetmacro{\r}{2.6} %

\tdplotsetmaincoords{80}{120}
\begin{tikzpicture}[
tdplot_main_coords,
font=\footnotesize,
Helpcircle/.style={gray!70!black,
},
]

\pgfmathsetmacro{\h}{0.9*\r} %

\pgfmathsetmacro{\t}{60}
\coordinate[label=left:{$\ket{\theta}$}] (X1) at ({\r*sin(\t)},0,{\r*cos(\t)});
\coordinate[label=left:{$\ket{\pi \!\! - \!\! \theta}$}] (X2) at ({\r*sin(\t)},0,{-\r*cos(\t)}); 

\coordinate (M) at (0,0,0);
\coordinate[label=$\ket{0}$] (Top) at (0,0,\r);
\coordinate[label=below:{$\ket{1}$}] (Bot) at (0,0,-\r);

\tdplotdrawarc{(M)}{\r}{-65}{110}{anchor=north}{}
\tdplotdrawarc[dashed]{(M)}{\r}{110}{295}{anchor=north}{}


\tdplotsetrotatedcoords{90}{90}{0}%
\tdplotdrawarc[tdplot_rotated_coords, blue]{(M)}{\r}{180}{360}{anchor=north}{}
\tdplotdrawarc[tdplot_rotated_coords, dashed]{(M)}{\r}{0}{180}{anchor=north}{}
\tdplotdrawarc[tdplot_rotated_coords]{(M)}{0.5*\r}{180}{180+\t}{anchor=north}{$\theta$}
\tdplotdrawarc[tdplot_rotated_coords]{(M)}{0.5*\r}{0}{-\t}{anchor=north}{$\theta$}
\draw[] (M) -- (X1);
\draw[] (M) -- (X2);
\draw[] (Top) -- (Bot);

\begin{scope}[tdplot_screen_coords, on background layer]
\fill[ball color= gray!20, opacity = 0.1] (M) circle (\r); 
\end{scope}

\foreach \P in {X1,X2}{
\shade[ball color=blue] (\P) circle (3pt);
}

\begin{scope}[-latex, shift={(M)}, xshift=1.5*\r cm, yshift=0.1*\r cm]
\foreach \P/\s/\Pos in {(1,0,0)/x/right, (0,1,0)/y/below, (0,0,1)/z/right} 
\draw[] (0,0,0) -- \P node[\Pos, pos=0.9,inner sep=2pt]{$\s$};
\end{scope}

\end{tikzpicture}
\caption{\label{fig:singleCutInitialization} A geometric depiction of the states $\ket{\theta}$ and $\ket{\pi-\theta}$ on the Bloch sphere; these two positions correspond to qubit positions found in at-$\theta$-warm-starts. The blue half-circle, $\textbf{Arc}$, in the $xz$-plane denotes all the possible positions for $\ket{\theta}$ as $\theta$ varies from $0$ to $\pi$. }
\end{figure}
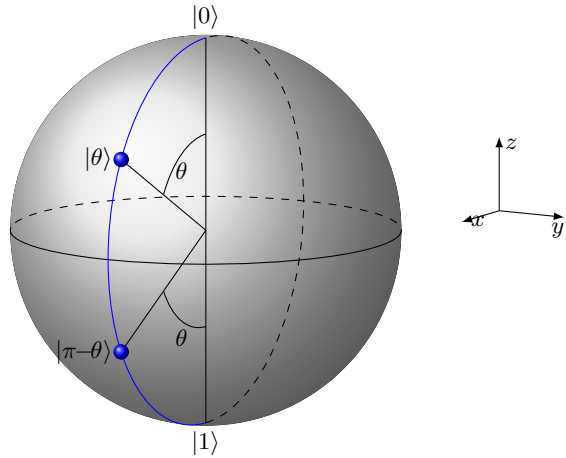

For the \MC{} problem, previous works have considered QAOA with at-$\theta$-warm-starts. Tate et al. \cite{tate2024guarantees} provided some theoretical results for at-$\theta$-warm-starts in the case of single-round Max-Cut QAOA on 3-regular graphs with aligned mixers. Egger et al. \cite{egger2021warm} consider a QAOA-variant (different than their QUBO-relaxation variant) with at-$\theta$-warm-starts with $\theta = \pi/3$, but with an \emph{unaligned} mixer that has the property that there exists parameters $\gamma$ and $\beta$ such that depth-1 QAOA yields the bitstring $b$ corresponding to the warm-start (i.e. it returns exactly $\ket{b}$). Cain et al. \cite{cain2022qaoa} considered Max-Cut QAOA with at-$\theta$-warm-starts with $\theta=0$ with the transverse field mixer; it was found that when the cut associated with the corresponding bitstring $b$ was a ``good" cut, then this variant of QAOA yields little to no improvement regardless of the circuit depth used.

\section{A Toy Problem}
\label{sec:toyProblem}
Before giving the formal result, we first consider a toy problem to give the reader intuition for why one might expect warm-started QAOA to perform poorly if the warm-start is initialized near the poles. This toy problem consists of a single qubit and the classical objective function is given by $c(x) = x$ for $x \in \{0,1\}$. This corresponding cost Hamiltonian is technically $C = \frac{1}{2}(I+Z)$ but we can instead\footnote{A similar replacement technique, obtained by shifting and scaling the cost objective, is discussed in Section \ref{sec:boundsOnRequiredDepth}.} use the simpler Hamiltonian $C = Z$.

For general problems, if the corresponding bitstring $b$ of a at-$\theta$-warm-start is non-optimal, then the QAOA circuit effectively needs to ``correct" certain qubits so that they are on the opposite hemisphere of the Bloch sphere from where they started. For our toy problem, since $\max_x c(x) = 1$, if we initialize with the warmstart $\ket{\psi_i} = \ket{\theta}$ with $\theta$ small, then the QAOA circuit needs to move the qubit from (near the) north pole of the Bloch sphere to the south pole.

For our toy problem, both QAOA unitaries, $U(B_{\ket{\psi_i}}, \beta)$ and $U(C, \gamma) = U(Z, \gamma)$ correspond to single-qubit rotations on the Bloch sphere, with the former rotating about the original qubit position and the latter rotating about the $z$-axis of the Bloch sphere. To simplify matters, we assume that the variational parameters $\gamma$ and $\beta$ are non-zero and are chosen so that the resulting (single-qubit) state stays in the $xz$-plane of the Bloch sphere after each QAOA unitary (i.e. rotation); this assumption puts an upper bound on the required circuit depth needed to achieve some change in approximation ratio since there may perhaps be better choices for $\gamma$ and $\beta$. When $\theta$ is small, it takes several layers of the QAOA circuit to significantly move the qubit away from its original position since the axes of rotation for both QAOA unitaries are so close to another as seen in Figure \ref{fig:toyExample}. In particular, the QAOA circuit evolves the state as follows:
\begin{equation}\ket{\theta} \to \ket{{-\theta}} \to \ket{{3\theta}} \to \ket{{-3\theta}} \to \ket{{5\theta}} \to \cdots .\end{equation}

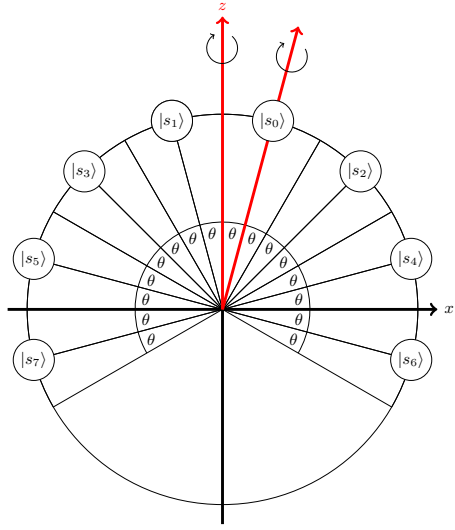
\begin{figure}[t]
\centering
\resizebox{0.5\linewidth}{!}{%
    \begin{tikzpicture}
      \def\n{8} 
      \def\r{3.8} 
    \def\a{1.5} 
    \def\b{120} 
    \def\m{20} 

    \draw[->,ultra thick] (-1.1*\r,0)--(1.1*\r,0) node[right]{$x$};
    \draw[-,ultra thick,black] (0,-1.1*\r)--(0,0) ;

      \draw (0,0) circle (\r);

      \pgfmathsetmacro{\startangle}{90+(-\n)*\b/\n}
      \pgfmathsetmacro{\endangle}{90+(\n)*\b/\n}
    \draw (\startangle:\a+0.2) arc (\startangle:\endangle:\a+0.2);
    
      \foreach \i in {1,...,\n} {
        \pgfmathsetmacro{\startangle}{90+(\i-1)*\b/\n}
        \pgfmathsetmacro{\endangle}{90+\i*\b/\n}
        \draw (0,0) -- (\startangle:\r);
        \draw (\startangle:\r) arc (\startangle:\endangle:\r);
        \draw (0,0) -- (\endangle:\r);
        \pgfmathsetmacro{\midangle}{\startangle+0.5*(\endangle-\startangle)}
        \node at (\midangle:\a) {$\theta$};
    }

      \foreach \i in {1,...,\n} {
        \pgfmathsetmacro{\startangle}{90-(\i-1)*\b/\n}
        \pgfmathsetmacro{\endangle}{90-\i*\b/\n}
        \draw (0,0) -- (\startangle:\r);
        \draw (\startangle:\r) arc (\startangle:\endangle:\r);
        \draw (0,0) -- (\endangle:\r);
        \pgfmathsetmacro{\midangle}{\startangle+0.5*(\endangle-\startangle)}
        \node at (\midangle:\a) {$\theta$};
    }

    \draw[->,ultra thick,red] (0,0)--(0,1.5*\r) node[above]{$z$};
    \draw[->] (88:1.4*\r) arc (50:-240:0.3);
    \pgfmathsetmacro{\axisangle}{90+(-1)*\b/\n}
    \draw[->,ultra thick,red] (0,0)--(\axisangle:1.5*\r);
    \draw[->] (\axisangle-2.6:1.4*\r) arc (30:-240:0.3);

  \foreach \i in {0,2,...,\n} {
    \pgfmathsetmacro{\lab}{int(\i-2)}
    \pgfmathsetmacro{\angle}{90-(\i-1)*\b/\n}
    \draw[black,fill=white] (\angle:\r) circle (0.4);
    \node at (\angle:\r) {$\ket{s_\lab}$};
  }

  \foreach \i in {2,4,...,\n} {
    \pgfmathsetmacro{\lab}{int(\i-1)}
    \pgfmathsetmacro{\angle}{90+(\i-1)*\b/\n}
    \draw[black,fill=white] (\angle:\r) circle (0.4);
    \node at (\angle:\r) {$\ket{s_\lab}$};
  }
    \end{tikzpicture}
}
\caption{\label{fig:toyExample} An illustration of the evolution of quantum state as each unitary of the QAOA circuit is applied. The large circle represents the intersection of the Bloch sphere with the $xz$-plane. The two alternating unitaries of the QAOA circuit correspond to a single-qubit rotation, represented by the two red arrows. Starting with the warm-start initial state $\ket{s_0}$, the state $\ket{s_j}$ ($j=0,1,\dots,7$) represents the quantum state after $j$ unitary operations of the QAOA; when $j$ is even, then $\ket{s_j}$ represents the result of a depth-$p$ QAOA circuit with $p=j/2$.}
\end{figure}

In general, after $p$ layers of the QAOA circuit, the resulting quantum state will be $\ket{{(2p+1)\theta}}$. Since $\max_x c(x)=1$, then the resulting initial and final approximation ratios are $\lambda_i = \sin^2(\theta / 2)$ and $\lambda_f = \sin^2((2p+1)\theta/2)$ respectively.

For very small values of $\theta$, $\lambda_i \approx 0$, and thus the change in approximation is approximately,
\begin{equation}\label{eqn:toyExample}\Delta \lambda \approx \lambda_f = \sin^2((2p+1)\theta/2).\end{equation}

If $\Delta \lambda$ is held constant, then the right-hand side of Equation \ref{eqn:toyExample} is also constant, which can only occur if $2p+1$ and $\theta$ are inversely proportional to one another, i.e., the required circuit depth scales as $p = O(1/\theta)$ for small $\theta$.

While the toy example gives a geometric intuition for why small initialization angles are bad, this intuition is difficult to directly generalize for multi-qubit systems since the Bloch sphere is not suitable for modeling entangled systems. Instead, we consider a more sophisticated approach in the next section, allowing us to obtain \emph{lower} bounds on the required circuit depth needed to achieve certain changes in approximation ratio.

\section{Bounds on Required Depth for Small Initialization Angle}
\label{sec:boundsOnRequiredDepth}

Recent work by Benchasattabuse et al. \cite{benchasattabuse2023lower} provide lower bounds on the number of rounds of QAOA that are required to obtain guaranteed approximation ratios; moreover, such bounds are applicable across a wide variety of optimization problems, phase separators, and mixing Hamiltonians. We now analyze such bounds in the context of warm-started QAOA and eventually show (Theorem \ref{thm:warmStartDepthBound}) that for a desired change of approximation ratio $\Delta\lambda$, the required number of rounds scales roughly as $\Omega(\Delta \lambda/\theta)$ for small $\theta$ approaching zero.

Benchasattabuse et al. \cite{benchasattabuse2023lower} obatined the QAOA circuit-depth bounds by utilizing one of the lower bounds on quantum annealing times found by Garc\'ia-Pintos et al. \cite{garcia2023lower} which we summarize below in Theorem \ref{thm:lowerBoundAnnealingTime}.

\begin{theorem}[Garc\'ia-Pintos et al. \cite{garcia2023lower}]\label{thm:lowerBoundAnnealingTime}
    Let $H(t) = (1-g(t))H_0 + g(t)H_1$ be a quantum annealing annealing schedule where $g(0) = 0$ and $g(t_f) = 1$, $t_f$ is the total annealing time, and $H_0$ and $H_1$ are Hamiltonians with zero ground state energies. For $t \in [0,t_f]$, let $\ket{\psi_t}$ be the quantum state at time $t$ of the annealing process. If $\ket{\psi_0}$ is a ground state of $H_0$, then the following inequality holds:
    \begin{equation}t_f \geq \frac{\langle H_0 \rangle_{t_f} + \langle H_1 \rangle_0 - \langle H_1 \rangle_{t_f}}{\Vert [H_1, H_0]\Vert},\end{equation}
    where for any appropriately-sized Hamiltonian $H$ and $t \in [0,t_f]$,
    $\langle H \rangle_t := \bra{\psi_t} H \ket{\psi_t}.$
\end{theorem}

Benchasattabuse et al. \cite{benchasattabuse2023lower} are able to translate these bounds on annealing times into bounds on QAOA circuit depth by using the fact that QAOA can be interpreted as quantum annealing with a ``bang-bang" schedule, i.e., $H_1$ and $H_0$ correspond to the phase separator and mixing Hamiltonian respectively, the annealing schedule is such that $g(t) \in \{0,1\}$ for all $t \in [0,t_f]$, and the annealing time is $\sum_{j=1}^p (|\beta_j|+|\gamma_j|)$.
\begin{remark}
    The above paragraph should \textbf{not} be misinterpreted to mean that the results of this paper only work for QAOA with specific parameter choices (e.g. choosing $\gamma$ and $\beta$ in a way corresponding to a linear annealing schedule). The results in this paper hold for \emph{any} choice of $\gamma$ and $\beta$ parameters.
\end{remark}

The QAOA circuit-depth bounds obtained by Benchasattabuse et al. (Theorem 2 of \cite{benchasattabuse2023lower}) assume that the initial state is the equal superposition state $\ket{+}^{\otimes n}$. It is not too difficult to adjust the bounds (and the corresponding proof) to account for arbitrary initial states. We state the adjusted bounds in Theorem \ref{thm:generalQAOADepthBounds} and prove that such an adjustment is correct.

\begin{theorem}
\label{thm:generalQAOADepthBounds}
    Given a classical objective function $c(x)$ for a maximization task, represented by the Hamiltonian $C$, encoded into a phase separator Hamiltonian $H_1 = c_\text{max} I - C$ and a mixing Hamiltonian $H_0$, where all Hamiltonians are $2\pi$ periodic with zero ground state energies. Let $c_\text{max}$ denote the global maximum of $c(x)$. Let $\ket{\psi_i}$ be a ground state of $H_0$ such that measurement of $\ket{\psi_i}$ yields an approximation ratio of $\lambda_i$.  If a QAOA protocol with $p$ rounds driven by $H_0$ and $H_1$ that starts from $\ket{\psi_i}$ and reaches a state $\ket{\psi_f}$ with approximation ratio $\lambda_f$, then,
    \begin{equation}p \geq \frac{\bra{\psi_f}H_0\ket{\psi_f} + \Delta\lambda \cdot c_\text{max}  }{4\pi \Vert [C, H_0]\Vert},\end{equation}
    where $\Delta \lambda = \lambda_f - \lambda_i$.
\end{theorem}
\begin{proof}
    In the proof by \cite{benchasattabuse2023lower}, they show that the above bound holds if $\Delta \lambda \cdot c_\text{max}$ is instead replaced with $\bra{\psi_i}H_1\ket{\psi_i} - \bra{\psi_f} H_1 \ket{\psi_f}$; this is true even if $\ket{\psi_i} \neq \ket{+}^{\otimes n}$. However, these two expressions are equal:
    \begin{align}\bra{\psi_i}H_1\ket{\psi_i} - \bra{\psi_f} H_1 \ket{\psi_f} &= \bra{\psi_f}C\ket{\psi_f} - \bra{\psi_i} C \ket{\psi_i} \\
    &= c_\text{max}\lambda_f- c_\text{max}\lambda_i \\
    &= \Delta \lambda \cdot c_\text{max},\end{align}
    hence showing that this adjusted version of the bound holds.

    The original bound,
    
    \begin{equation}p \geq \frac{\bra{\psi_f}H_0\ket{\psi_f} + c_\text{max}\lambda_f - c_\text{avg}  }{4\pi \Vert [C, H_0]\Vert}.\end{equation}
    
    in the statement of Theorem 2 of \cite{benchasattabuse2023lower}, which assumes that $\ket{\psi_i} = \ket{+}^{\otimes n}$, can be obtained from our adjusted bounds by observing that, \begin{equation}\lambda_i c_\text{max} = \frac{\bra{+}^{\otimes n}C\ket{+}^{\otimes n}}{c_\text{max}}c_\text{max} = c_\text{avg}.\end{equation}
\end{proof}

For convenience, we consider a looser-version of the bound in Corollary \ref{thm:looserDepthBound} and use $p_\text{min}$ to denote the value of this looser bound.

\begin{corollary}
\label{thm:looserDepthBound}
    The bound in the statement of Theorem \ref{thm:generalQAOADepthBounds} can be replaced with
    \begin{equation}p \geq p_\text{min} := \frac{\Delta \lambda \cdot c_\text{max}}{4\pi \Vert [C, H_0]\Vert}.\end{equation}
\end{corollary}
\begin{proof}
    This follows from Theorem \ref{thm:generalQAOADepthBounds} and that $\bra{\psi_f} H_0 \ket{\psi_f} \geq 0$ (as $H_0$ is assumed to have a zero ground state energy).
\end{proof}

Before continuing, we point out certain important conditions in Theorem \ref{thm:generalQAOADepthBounds}. First, all of the mixers need to be $2\pi$ periodic. In particular, this will hold for the cost Hamiltonian if the corresponding classical objective function $c: \{0,1\}^n \to \mathbb{R}$ has integer objective values, i.e., $c(x) \in \mathbb{Z}$ for all $x \in \{0,1\}^n$. Now consider the transverse field mixer $B_\text{TF}$ and a warm-start mixer $B_{\ket{\psi_i}}$ that is aligned with the warm-start $\ket{\psi_i}$. It can be shown that both of these are $2\pi$ periodic; however, Theorem \ref{thm:generalQAOADepthBounds} also requires that the mixers also have zero energy which is not the case for $B_\text{TF}$ and $B_{\ket{\psi_i}}$. In particular, for \emph{any} product state $\ket{s}$, the mixer $B_{\ket{s}}$ has energies ranging from $-n$ to $n$ (see \cite{Tate2023warmstartedqaoa} for a proof); however, we can shift the mixers to have energies ranging from $0$ to $n$ as so: $\hat{B}_{\ket{s}} := \frac{1}{2}(nI - B_{\ket{s}}).$
Choosing $\ket{s}$ appropriately, we have
$\hat{B}_\text{TF} := \frac{1}{2}(nI -B_\text{TF})$ and  $\hat{B}_{\ket{\psi_i}} := \frac{1}{2}(nI - B_{\ket{\psi_i}}),$
where $\hat{B}_\text{TF}$ and $\hat{B}_{\ket{\psi_i}}$ have energies ranging from 0 to $n$. Moreover, now $\ket{+}^{\otimes n}$ and $\ket{\psi_i}$ are \emph{ground} states of $\hat{B}_\text{TF}$ and $B_{\ket{\psi_i}}$ respectively as required by Theorem \ref{thm:generalQAOADepthBounds}. With some appropriate rescaling of the variational parameters $\gamma$ and $\beta$, one can show that replacing $B_\text{TF}$ with $\hat{B}_\text{TF}$ (or $B_{\ket{\psi_i}}$ with $\hat{B}_{\ket{\psi_i}}$) in the QAOA circuit will result in the same final state. Nonetheless, it is important to consider the versions of these mixers with zero-ground state energy as the commutator term $[C, H_0]$ in Theorem \ref{thm:generalQAOADepthBounds} and Corollary \ref{thm:looserDepthBound} assume the mixer is written in such a form. The following lemma will be helpful in regards to analyzing such commutators.


\begin{lemma}
\label{thm:commutatorRelation}
    Let $\ket{\psi_i} = \bigotimes_{j=1}^n \ket{\theta_j} = \bigotimes_{j=1}^n \ket{\vec{n}_j}$ be a within-$\theta$-warm-start. Then, $[B_{\vec{n}_j,j}, C] = \sin(\hat{\theta}_j)[X_j, C],$ where $\hat{\theta_j}$ is the angle that the $j$th qubit is away from one of the poles of the Bloch sphere (defined in Section \ref{sec:constructionOfWarmStartedStates}).
\end{lemma}
\begin{proof}
     Recall that $B_{\ket{\psi_i}} = \sum_{j=1}^n B_{\vec{n}_j,j},$
    where
    $B_{\vec{n}_j, j} = x_jX_j+y_jY_j+z_jZ_j,$
    where $\vec{n}_j = (x_j, y_j, z_j)$ is the position of the $j$th qubit in Cartesian coordinates. Given $\ket{\psi_i}=\bigotimes_{j=1}^n \ket{\theta_j}$, we can calculate the $j$th qubit position as:
    \begin{equation}\vec{n}_j = (\sin(\theta_j), 0, \cos(\theta_j)) 
    \end{equation}
    and thus, we have that, $B_{\vec{n}_j} = \sin(\theta)X_j + \cos(\theta)Z_j.$ Since $[Z_j, C] = 0$ for all $j\in [n]$, we have that $[B_{\vec{n}_j,j}, C] = \sin(\theta_j)[X_j, C].$
    
    To complete the proof, it suffices to show that $\sin(\theta_j) = \sin(\hat{\theta}_j)$; this can be seen by recalling the identity $\sin(x) = \sin(\pi-x)$ and the fact that $\hat{\theta_j} = \min(\theta_j, \pi-\theta_j)$.
\end{proof}

 Using the commutator relation in Lemma \ref{thm:commutatorRelation}, we now present our main result which places a lower bound on the required number of angles for warm-started QAOA with within-$\theta$-warm-starts.

\begin{theorem}
\label{thm:warmStartDepthBoundVaryingAngles}
    Fix a phase separator $H_1 = c_\text{max}I -C$ (where $C$ is a Hamiltonian that corresponds to some classical objective function $c(x)$), choice of $\Delta \lambda$, and a within-$\theta$-warm-start $\ket{\psi_i}$. Let  $p_\text{min}(\hat{B}_{\ket{\psi_i}})$ correspond to the bound $p_\text{min}$ in Corollary \ref{thm:looserDepthBound} where $H_0$ is replaced with $H_0 = \hat{B}_{\ket{\psi_i}}$. Then,
    \begin{equation}p_\text{min}(\hat{B}_{\ket{\psi_i}}) \geq \frac{\Delta \lambda}{\sin (\theta)} \mathcal{F}(c),\end{equation}
    where $\mathcal{F}(c)$ is a function that depends on the classical objective $c$ and is independent of $\Delta \lambda, \ket{\psi_i},$ and $\theta$.
\end{theorem}
\begin{proof}

    Recall from Remark \ref{remark:zeroPhase}, that we may assume (without loss of generality) that $\ket{\psi_i}$ has the form $\ket{\psi_i} = \bigotimes_{j=1}^n \ket{\theta_j}$ where $\hat{\theta_j}$ (the angle that the $j$th qubit of $\ket{\psi_i}$ is from the nearest pole on the Bloch sphere) lies in the interval $[0,\theta]$.
     
     Now, recall that $B_{\ket{\psi_i}} = \sum_{j=1}^n B_{\vec{n}_j,j},$
     where $\vec{n}_j = (x_j, y_j, z_j)$ is the position of the $j$th qubit in Cartesian coordinates. Also recall from Lemma \ref{thm:commutatorRelation} that
     $[B_{\vec{n}_j,j}, C] = \sin(\hat{\theta}_j)[X_j, C]$. Thus, it follows that,
    \begin{equation} [B_{\ket{\psi_i}}, C] = \sum_{j=1}^n [B_{\vec{n}_j,j}, C] 
    =  \sum_{j=1}^n \sin(\hat{\theta}_j) [X_j, C].\end{equation}

    Taking norms on both sides yields,
    $$\Vert [B_{\ket{\psi_i}}, C] = \left\Vert \sum_{j=1}^n \sin(\hat{\theta}_j) [X_j, C] \right\Vert \leq \sum_{j=1}^n \sin(\theta)\left\Vert [X_j, C] \right\Vert = \frac{\sin(\theta)}{\mathcal{F}(c)} \cdot \frac{c_\text{max}}{2\pi},$$

    where,
    \begin{equation}
        \mathcal{F}(c) = \frac{c_\text{max}}{2\pi} \cdot \left(\sum_{j=1}^n \left\Vert [X_j, C] \right\Vert\right)^{-1}.
    \end{equation}

    Using basic commutator properties, it is easy to show that $[\hat{B}_{\ket{\psi_i}}, C] = -\frac{1}{2}[B_{\ket{\psi_i}}, C]$. Finally, from the definition of $p_\text{min}$ in Corollary \ref{thm:looserDepthBound}, we have
    
        $$p_\text{min}(\hat{B}_{\ket{\psi_i}}) = \frac{\Delta \lambda \cdot c_\text{max}}{4\pi \Vert [C, \hat{B}_{\ket{\psi_i}}]\Vert} 
        = \frac{\Delta \lambda \cdot c_\text{max}}{2\pi \Vert [C, B_{\ket{\psi_i}}]\Vert}  
        \geq
         \frac{\Delta \lambda}{\sin (\theta)} \mathcal{F}(c).$$
    
\end{proof}




Note that for small $\theta$, $\sin(\theta) \approx \theta$, so in such a case, the required number of rounds of warm-started QAOA, as a function of initialization angle $\theta$ and desired change in approximation ratio $\Delta \lambda$, scales as $p = \Omega(\Delta \lambda/\theta)$. In particular, the number of required rounds approaches infinity as $\theta$ approaches 0, which suggests that the method of warm-started QAOA discussed in this work is not recommended for very small choices of $\theta$ if one wishes to obtain an appreciable value of $\Delta \lambda$. 

In the special case of at-$\theta$-warm-starts, where each qubit in the warm-start is \emph{exactly} angle $\theta$ away from one of the poles, we have a slightly stronger result that directly relates the lower bounds on the required number of rounds required for warm-started QAOA vs standard QAOA.

\begin{theorem}
\label{thm:warmStartDepthBound}
    Fix a phase separator $H_1 = c_\text{max}I -C$ (where $C$ is a Hamiltonian that corresponds to some classical objective function $c(x)$), choice of $\Delta \lambda$, and let $\ket{\psi_i}$ be an at-$\theta$-warm-start for some $0 \leq \theta \leq \pi/2$. Let $p_\text{min}(\hat{B}_\text{TF})$ and $p_\text{min}(\hat{B}_{\ket{\psi_i}})$ correspond to the bound $p_\text{min}$ in Corollary \ref{thm:looserDepthBound} where $H_0$ is replaced with $H_0 = \hat{B}_\text{TF}$ and $H_0 = \hat{B}_{\ket{\psi_i}}$ respectively. Then,
    \begin{equation}p_\text{min}(\hat{B}_{\ket{\psi_i}}) = \frac{1}{\sin \theta} \cdot p_\text{min}(\hat{B}_\text{TF}).\end{equation}
\end{theorem}
\begin{proof}
    Recall from Remark \ref{remark:zeroPhase}, that we may assume (without loss of generality) that $\ket{\psi_i}$ has the form $\ket{\psi_i} = \bigotimes_{j=1}^n \ket{\theta_j}$. Since $\ket{\psi_i}$ was an at-$\theta$-warm-start, then $\hat{\theta_j}$ (the angle that the $j$th qubit of $\ket{\psi_i}$ is from the nearest pole on the Bloch sphere) is equal to $\theta$ for all $j\in[n]$.

    As a result of Lemma \ref{thm:commutatorRelation} we have
     $[B_{\vec{n}_j,j}, C] = \sin(\hat{\theta}_j)[X_j, C] = \sin(\theta)[X_j,C]$. Thus, it follows that,
    \begin{equation} [B_{\ket{\psi_i}}, C] = \sum_{j=1}^n [B_{\vec{n}_j,j}, C] 
    =  \sin(\theta)\sum_{j=1}^n [X_j, C] = \sin(\theta) \cdot [B_\text{TF}, C].\end{equation}
    The result then follows from the commutator relation above and the definition of $p_\text{min}$ in Corollary \ref{thm:looserDepthBound}.
\end{proof}

    One should be careful in interpreting Theorem \ref{thm:warmStartDepthBound}; since $p_\text{min}(\hat{B}_{\ket{\psi_i}}) \geq p_\text{min}(\hat{B}_\text{TF})$, at first glance, it appears that it is always preferable to run standard QAOA. However, this is not necessarily the case. First, the inequality is a result of comparing lower bounds which are not necessarily tight, the relationship between the actual number of required rounds (between standard QAOA and warm-started QAOA) may actually be much different. Second, even if the bounds were tight, this relationship between $p_\text{min}(\hat{B}_{\ket{\psi_i}})$ and $p_\text{min}(\hat{B}_\text{TF})$  assumes that $\Delta \lambda$ fixed; if instead the target approximation ratio $\lambda_f$ is fixed and $\lambda_i$ is allowed to vary, then this relationship does not necessarily hold. In particular, numerical evidence such as those found in \cite{Tate2023warmstartedqaoa} and \cite{egger2021warm}, suggest that for many instances and for a suitable choice of $\theta$, the number of required rounds to achieve a particular target approximation ratio $\lambda_f$ for warm-started QAOA is much less compared to standard QAOA.



    \section{Discussion and Conclusion}
    \label{sec:discussionConclusion}

    This work is consistent with preliminary numerical simulations \cite{Tate2023warmstartedqaoa,egger2021warm}  which suggests that there is very little change in the approximation ratio when the initialization angle $\theta$ is near zero; there were previous geometrical intuitions for why this was the case, but this work shows definitively that small choices of initialization angles are not suitable for warm-started QAOA with aligned mixers. 
    

    It should be noted our results are fundamentally a consequence of QAOA with \emph{aligned} mixers being able to be interpreted as a quantum-annealing protocol with a bang-bang schedule; for QAOA with non-aligned mixers, it is likely the case that an entirely different type of analysis is necessary to obtain similar lower bounds on the circuit depth.

\begin{credits}
\subsubsection{\ackname} This work was supported by the U.S. Department of Energy through the Los Alamos National Laboratory.
Los Alamos National Laboratory is operated by Triad National Security, LLC, for the National Nuclear
Security Administration of U.S. Department of Energy (Contract No. 89233218CNA000001). 

 \end{credits}
%
%
%
\bibliographystyle{splncs04}
\bibliography{SOFSEM}

\begin{thebibliography}{10}
\providecommand{\url}[1]{\texttt{#1}}
\providecommand{\urlprefix}{URL }
\providecommand{\doi}[1]{https://doi.org/#1}

\bibitem{benchasattabuse2023lower}
Benchasattabuse, N., B{\"a}rtschi, A., Garc{\'\i}a-Pintos, L.P., Golden, J.,
  Lemons, N., Eidenbenz, S.: Lower bounds on number of qaoa rounds required for
  guaranteed approximation ratios. arXiv preprint arXiv:2308.15442  (2023)

\bibitem{binkowski2023elementary}
Binkowski, L., Ko{\ss}mann, G., Ziegler, T., Schwonnek, R.: Elementary proof of
  qaoa convergence. arXiv preprint arXiv:2302.04968  (2023)

\bibitem{blekos2024review}
Blekos, K., Brand, D., Ceschini, A., Chou, C.H., Li, R.H., Pandya, K., Summer,
  A.: A review on quantum approximate optimization algorithm and its variants.
  Physics Reports  \textbf{1068},  1--66 (2024)

\bibitem{cain2022qaoa}
Cain, M., Farhi, E., Gutmann, S., Ranard, D., Tang, E.: The qaoa gets stuck
  starting from a good classical string. arXiv preprint arXiv:2207.05089
  (2022)

\bibitem{egger2021warm}
Egger, D.J., Mare{\v{c}}ek, J., Woerner, S.: Warm-starting quantum
  optimization. Quantum  \textbf{5}, ~479 (2021)

\bibitem{farhi2014quantum}
Farhi, E., Goldstone, J., Gutmann, S.: A quantum approximate optimization
  algorithm. arXiv preprint arXiv:1411.4028  (2014)

\bibitem{garcia2023lower}
Garc{\'\i}a-Pintos, L.P., Brady, L.T., Bringewatt, J., Liu, Y.K.: Lower bounds
  on quantum annealing times. Physical Review Letters  \textbf{130}(14),
  140601 (2023)

\bibitem{blochSphereRotations}
Glendinning, I.: Rotations on the bloch sphere (2010).
  \doi{10.13140/RG.2.2.27566.25922}

\bibitem{HWORVB17}
Hadfield, S., Wang, Z., O’Gorman, B., Rieffel, E.G., Venturelli, D., Biswas,
  R.: From the quantum approximate optimization algorithm to a quantum
  alternating operator ansatz. Algorithms  \textbf{12}(2) (2019).
  \doi{10.3390/a12020034}, \url{https://www.mdpi.com/1999-4893/12/2/34}

\bibitem{tate2024guarantees}
Tate, R., Eidenbenz, S.: Guarantees on warm-started qaoa: Single-round
  approximation ratios for 3-regular maxcut and higher-round scaling limits.
  arXiv preprint arXiv:2402.12631  (2024)

\bibitem{tate2023bridging}
Tate, R., Farhadi, M., Herold, C., Mohler, G., Gupta, S.: Bridging classical
  and quantum with sdp initialized warm-starts for qaoa. ACM Transactions on
  Quantum Computing  \textbf{4}(2),  1--39 (2023)

\bibitem{Tate2023warmstartedqaoa}
Tate, R., Moondra, J., Gard, B., Mohler, G., Gupta, S.: Warm-{S}tarted {QAOA}
  with {C}ustom {M}ixers {P}rovably {C}onverges and {C}omputationally {B}eats
  {G}oemans-{W}illiamson's {M}ax-{C}ut at {L}ow {C}ircuit {D}epths. {Quantum}
  \textbf{7}, ~1121 (Sep 2023). \doi{10.22331/q-2023-09-26-1121},
  \url{https://doi.org/10.22331/q-2023-09-26-1121}

\end{thebibliography}
\end{document}